\documentclass[11pt]{article}

\usepackage[margin=1in]{geometry}
\usepackage{amsmath,amssymb,amsthm,mathtools}
\usepackage[T1]{fontenc}
\usepackage{microtype}
\microtypesetup{expansion=false}
\usepackage{enumitem}
\usepackage{xcolor}
\usepackage{booktabs}
\usepackage{tikz,pgfplots}
\pgfplotsset{compat=1.18}
\usepgfplotslibrary{fillbetween}
\usetikzlibrary{positioning,fit,calc,arrows.meta}
\usepackage{hyperref}
\usepackage[nameinlink,capitalize,noabbrev]{cleveref}
\usepackage{aliascnt}
\usepackage[normalem]{ulem}

\definecolor{linkgreen}{RGB}{0,90,160}
\definecolor{linkblue}{RGB}{0,75,180}
\definecolor{accent}{RGB}{0,120,105}
\hypersetup{colorlinks=true, linkcolor=linkblue, citecolor=linkblue, urlcolor=linkblue}

\newtheorem{theorem}{Theorem}[section]
\newaliascnt{lemma}{theorem}
\newtheorem{lemma}[lemma]{Lemma}
\aliascntresetthe{lemma}
\newaliascnt{proposition}{theorem}
\newtheorem{proposition}[proposition]{Proposition}
\aliascntresetthe{proposition}
\newaliascnt{corollary}{theorem}

\aliascntresetthe{corollary}
\theoremstyle{definition}
\newaliascnt{definition}{theorem}
\newtheorem{definition}[definition]{Definition}
\aliascntresetthe{definition}
\newaliascnt{remark}{theorem}

\aliascntresetthe{remark}

\crefname{theorem}{Theorem}{Theorems}\Crefname{theorem}{Theorem}{Theorems}
\crefname{lemma}{Lemma}{Lemmas}\Crefname{lemma}{Lemma}{Lemmas}
\crefname{proposition}{Proposition}{Propositions}\Crefname{proposition}{Proposition}{Propositions}
\crefname{corollary}{Corollary}{Corollaries}\Crefname{corollary}{Corollary}{Corollaries}
\crefname{definition}{Definition}{Definitions}\Crefname{definition}{Definition}{Definitions}
\crefname{remark}{Remark}{Remarks}\Crefname{remark}{Remark}{Remarks}

\newcommand{\st}{\mathsf{st}}
\newcommand{\yes}{\textup{\textsc{yes}}}
\newcommand{\no}{\textup{\textsc{no}}}
\newcommand{\Dyes}{\mathcal{D}_{\mathsf{y}}}
\newcommand{\Dno}{\mathcal{D}_{\mathsf{n}}}
\newcommand{\dtv}{d_{\mathrm{TV}}}
\newcommand{\trans}{\mathsf{trans}}
\newcommand{\eps}{\varepsilon}
\newcommand{\nb}{\mathrm{nb}}
\newcommand{\touch}{\mathrm{Touch}}
\newcommand{\tr}{\mathrm{tr}}
\DeclareMathOperator*{\Ex}{\mathbb{E}}
\DeclareMathOperator{\Prob}{\mathbb{P}}

\title{A Better Lower Bound for Sublinear Metric Steiner Tree}
\author{Sepideh Mahabadi\thanks{Microsoft Research. Email: \texttt{smahabadi@microsoft.com}.}
  \and
  Jakub Tarnawski\thanks{Microsoft Research. Email: \texttt{jatarnaw@microsoft.com}.}
  \and
  Ali Vakilian\thanks{Rutgers University. Email: \texttt{a.vakilian@rutgers.edu}. Supported in part by NSF CAREER Award CCF-2542925.}
}
\date{}

\begin{document}
\maketitle

\begin{abstract}
In the metric Steiner tree problem, we are given an $n$-point metric space $(V,w)$ and a set of terminals $T\subseteq V$, and we are interested in the smallest total weight $\st(V,T,w)$ of a tree that contains every terminal and may also use points of $V\setminus T$.  We work in the distance-oracle model: the algorithm knows $V$ and $T$ in advance and learns $w$ only by querying a single entry $w(u,v)$ at a time. The goal is to estimate $\st(V,T,w)$ with the minimum number of queries.

A factor of $2$ is achievable with nearly linear number of queries, so the interesting question is the query complexity of algorithms with an approximation ratio strictly below $2$.  We prove that for every fixed $\eta\in(0,1)$, every randomized $(2-\eta)$-approximation algorithm must make more than $\widetilde{\Omega}(n^{7/5})$ queries, improving upon the previously known bound of $\widetilde{\Omega}(n^{6/5})$.  Given that the best known upper bound shows that $\widetilde O(n^{5/3})$ queries suffice, the exponent of this problem now lies between $7/5$ and $5/3$, up to logarithmic factors.
\end{abstract}

\section{Introduction}

The {\em Steiner tree} problem asks for a minimum-cost network connecting a prescribed set of terminals, with the option of using additional nonterminal vertices to reduce the cost. It is one of the oldest problems in combinatorial optimization and one of Karp's original NP-complete problems~\cite{Karp72}. 
It is also among the most well-studied problems in combinatorial optimization: the minimum spanning tree on the terminals is a $2$-approximation for Steiner tree, and the best ratio known is $\ln 4+\eps\approx 1.39$~\cite{BGRS13}. Further, no polynomial-time algorithm can achieve a ratio better than $96/95$ unless $\mathrm{P}=\mathrm{NP}$~\cite{CC08}. 
Steiner tree and its variants have also been studied in a broad range of computational models beyond the standard offline setting, including streaming~\cite{CJKV22},
dynamic~\cite{IW91,GK14,LOPSZ15}, online~\cite{IW91,GuGK13,AAB96,BC97,BDM22}, and sublinear-query
models~\cite{CKT23,MRTV25}.

\smallskip
An instance of the \emph{metric Steiner tree} problem consists of a finite set of points $V$, a
metric $w$ on $V$, and a set of \emph{terminals} $T\subseteq V$.  A feasible solution is a tree in
the complete graph on $V$ whose vertex set contains $T$, and the points of $V\setminus T$ that it can optionally use.  We write $\st(V,T,w)$ for the smallest weight of a feasible solution, and throughout
we put $n=|V|$ and $k=|T|$.

\smallskip
We study the problem in the \emph{distance-oracle} model: an algorithm may query any unordered pair
$\{u,v\}\subseteq V$ and is told $w(u,v)$, and its cost is the number of such queries.  The input
has $\binom n2$ entries, so the question of interest is what can be done with far fewer than $n^2$
queries.

A factor of $2$ is the natural baseline in this model. Indeed, the minimum spanning tree of the metric restricted to $T$ has weight between $\st(V,T,w)$ and $2\,\st(V,T,w)$, and hence gives a $2$-estimation if computed exactly. Moreover, by a result of Czumaj and Sohler~\cite{CS09}, its weight can be estimated within a factor of $1+\eps$ using $\widetilde O(k)$ queries.
Consequently, a $(2+\eps)$-estimation is achievable with near-linear query complexity. The natural question is therefore what the best achievable running time or query complexity is for obtaining an approximation factor strictly below $2$.

Chen, Khanna and Tan~\cite{CKT23} began this line of work and settled most of the landscape.  They
showed that getting within a factor $5/3-\eps$ needs $\Omega_\eps(n^2)$ queries even when
$k=O(1)$. They further showed that \emph{returning} a tree of weight strictly less than $2\,\st(V,T,w)$ needs $\Omega(nk)$
queries, matching an $O(nk)$ upper bound. Finally, they showed that some constant factor
$2-\eps_0$ can be \emph{estimated} with $\widetilde O(n^{13/7})=o(n^2)$ queries.  Mahabadi,
Roghani, Tarnawski and Vakilian~\cite{MRTV25} improved the last bound to $\widetilde O(n^{5/3})$.
For lower bounds on estimation below $2$, the only previously known result is that of
\cite{CKT23}: $\widetilde\Omega(n+k^{6/5})$ queries are needed, which reads as
$\widetilde\Omega(n^{6/5})$ when $k=\Theta(n)$.
Estimating below the factor $2$ is thus the one task in this landscape whose complexity remains open.

The metric Steiner tree belongs to a broader line of work on graph optimization problems in the distance-oracle model, including minimum spanning tree (MST)~\cite{chazelle2005approximating,CS09}, maximal independent set (MIS)~\cite{nguyen2008constant,yoshida2009improved,mahabadi2026sublinear}, matching~\cite{yoshida2009improved,kapralov2020space, Behnezhad21,BehnezhadRR23a,BehnezhadRRS-SODA23,azarmehr2025tight,MahabadiRT25,behnezhad2024approximating,bhattacharya2023sublinear}, metric Steiner Forest~\cite{mahabadi2026sublinear}, and metric TSP~\cite{CKK20,chen2023sublinear}.

\subsection{Our results}
We now state our result and place it in the known query-complexity landscape for metric Steiner tree. We show that, when $k=\Theta(n)$, any fixed improvement over factor $2$ requires $\Omega(n^{7/5}/\log n)$ queries.  
Throughout, we consider the estimation task. 

\begin{theorem}\label{thm:main}
For every constant $\eta\in(0,1)$ there are constants $c_\eta>0$ and $n_\eta$ such that the
following holds for every $n\ge n_\eta$.  There is a distribution over $n$-point metric Steiner tree
instances $(V,T,w)$, with $|T|=\Theta_\eta(n)$ and with every distance between two distinct points
lying in $\{\Theta_\eta(n^{-3}),1,2\}$, on which every deterministic algorithm that returns a number
$\widehat C$ with
\[
  \st(V,T,w)\ \le\ \widehat C\ \le\ (2-\eta)\,\st(V,T,w).
\]
and making at most $c_\eta n^{7/5}/\log n$ queries fails with probability greater than $1/3$.
Consequently, every randomized algorithm that returns such a number with probability at least $2/3$
on every instance makes more than $c_\eta n^{7/5}/\log n$ queries.
\end{theorem}

\begin{table}[t]
\centering
\small
\begin{tabular}{@{}llll@{}}
\toprule
Problem & Lower bound & Upper bound & Reference\\
\midrule
Estimation within $2+\eps$ & $\widetilde\Omega(k)$ & $\widetilde O(k)$ & \cite{CS09}\\
Estimation within $2-\eps$, $\eps<\tfrac13$ & $\widetilde\Omega(n+k^{6/5})$ & $\widetilde O(n^{5/3})$
  & \cite{CKT23,MRTV25}\\
\quad\emph{\textbf{this paper}}, $k=\Theta_\eps(n)$
  & $\boldsymbol{\Omega_\eps(n^{7/5}/\log n)}$ & & \cref{thm:main}\\[2pt]

Estimation within $\tfrac53-\eps$ & $\Omega_\eps(n^2)$ & $\Theta(n^2)$ & \cite{CKT23}\\
\midrule
Return a tree of weight $\le(2-\eps)\st$, $\eps<\tfrac13$ & $\Omega(nk)$ & $O(nk)$ & \cite{CKT23}\\
Return a tree of weight $\le\alpha\,\st$, constant $\alpha\ge2$ & $\Omega(k^2/\alpha)$
  & $\widetilde O(k^2/\alpha)$ & \cite{CKT23}\\
\bottomrule
\end{tabular}
\caption{Known bounds for metric Steiner tree in the distance-oracle model, and the line that
\cref{thm:main} improves. The bound $\widetilde O(n^{5/3})$ holds for one particular constant
$\eps_0>0$; \cref{thm:main} holds for every fixed $\eta$, in particular for that $\eps_0$.}\label{tab:prior}
\end{table}

We remark that following \cite{CKT23,MRTV25} we require the algorithm to output a value that is at least the optimum and at
most $(2-\eta)$ times the optimum.  The proof uses only that the interval of acceptable answers has
its two endpoints in the ratio $2-\eta$.  

Moreover, although \cref{thm:main} is stated for worst-case query complexity, the same lower bound holds, up to a constant factor, for expected query complexity.  Indeed, after constant-probability amplification, an algorithm using $q$ queries in expectation can be truncated after a sufficiently large constant multiple of $q$ queries, and Markov's inequality decreases its success probability by only a constant. The same argument then applies to the truncated algorithm.

\subsection{Overview of the proof}\label{sec:overview}

\paragraph{The two distributions.}
Fix a constant integer $d=\Theta(1/\eta)$ and a growing integer $m$, and set $A=m^3$, $R=m^2$, and $K=AR=m^5$. 

The point set is partitioned into {\em terminals} and {\em non-terminals}.
There are $K$ \emph{ordinary terminals}, split into $A$ hidden \emph{clusters} of $R$ each. We set the distance of any two
ordinary terminals of one cluster to be $\lambda=K^{-3}$, so a cluster is a tight group
that costs almost nothing to connect up.  There are $dA$ \emph{special terminals}, forming a hidden
uniformly random subset of the terminal labels, so an algorithm does not know which terminals are
ordinary and which are special.  Finally there are $R$ public \emph{blocks} of non-terminal points,
which we call \emph{centers}.

Independently in each block, the special terminals are split into $A$ hidden \emph{groups} of $d$
each,
and each group which we denote by $U$ is given $b=(d-1)^2$ of the block's centers, a publicly announced row of the block. What is hidden is the order of centers within that row.  A center is
at distance $1$ from a certain subset of its row's group of terminals $U$, and at distance $2$ from everything else, and further all distances not yet named are $2$.  The two distributions differ only in these subsets:

\begin{itemize}[leftmargin=1.6em,itemsep=2pt]
  \item In the \no{} distribution, the $b$ centers of a group $U$'s row contain the $\binom d2$ two-element
        subsets of $U$, together with $\tfrac{(d-1)(d-2)}2$ empty subsets.  Every center is then at
        distance $1$ from at most two terminals, so no center is worth using.
  \item In the \yes{} distribution one hidden block $g$ is different.  There the subsets contain $U$ itself once, and $d-2$ copies of $\{u\}$ for each $u\in U$.  Note that in this case, one center per group is now at distance $1$ from all $d$ members of the group, and joining a group through that center costs $d$ instead of the $2(d-1)$ that $d$ separate points would cost.
\end{itemize}

Both families use exactly $b=(d-1)^2$ centers, so the two distributions live on point sets of the
same size and the accounting matches.

\paragraph{The gap in cost.}
In a \no{} instance no center helps, so the optimum contracts each cluster at essentially no cost
and then pays $2$ for each of the $(d+1)A$ remaining pieces, for a total of about $2(d+1)A$.  In a
\yes{} instance each of the $A$ groups can additionally be gathered into a star for $d$ units,
leaving only $2A$ pieces, for a total of about $(d+4)A$.  The ratio tends to
\[
  \frac{2(d+1)}{d+4}=2-\frac{6}{d+4},
\]
which is above $2-\eta$ once $d>6/\eta-4$.  This is where the constant $d=\Theta(1/\eta)$ comes
from, and why the ratio approaches $2$ without reaching it.

\paragraph{Why the difference is hard to see.}
The two families of centers are built to be indistinguishable locally.  In both, each group's row has
exactly $b$ centers, each terminal of the group has exactly $d-1$ of them at distance $1$, and the
group accounts for exactly $d(d-1)$ terminal--center pairs at distance $1$.  \Cref{lem:trade} says
more: for any at most $\alpha = 2$  terminals of a group there is a one-to-one correspondence between the two
families that preserves what those  terminals see.  A difference appears only for $\alpha+1=3$
terminals of one group, since in the \yes{} family one center is at distance $1$ from all three of them
and in the \no{} family no center ever is.

So call a group \emph{exposed} if the algorithm has queried three of its members against centers
in this very group's row, and a block exposed if it contains an exposed group.  \Cref{lem:coupling}
turns the correspondence into a coupling argument and shows that the two distributions of query-answer
transcripts differ by at most the probability that the hidden block $g$ is exposed.

\paragraph{Counting exposed groups.}
Two things must be counted.  First, how many special terminals the algorithm can single out at all.
A special terminal is \emph{noticed} if either it ever gets the answer $1$, or if it is queried many times, or if the
algorithm has poked around its centers a great deal.  Two exchange arguments
(\cref{lem:cluster,lem:special}) show that these events are rare: conditioned on everything seen so
far, a fresh query returns $1$ at a given terminal with probability $O_d(1/K)$, and a given terminal
is special with probability $O_d(1/R)$, which are roughly equal to their unconditional probabilities.  Feeding this back into the
run gives $\Ex[|W|]=O_d(q/K)$ for the set $W$ of noticed terminals (\cref{lem:W}).

Second, how often a set that small contains three members of a common group.  Since $W$ is chosen by
the algorithm using everything it has seen, we need a bound valid for \emph{all} sets at once.
\Cref{lem:sparsity} supplies one: with probability $1-A^{-10}$, every set $Q$ of at most $R$ special
terminals meets three-in-a-group in at most $O(|Q|\log A)$ of the $RA$ block--group pairs.  Together
with a bound on lucky queries that land on an unnoticed terminal at one of its own group's centers,
this gives
\[
  \Prob[\,g\text{ exposed}\,]=O_d\!\left(\frac{q\log A}{KR}+\frac q{AK}\right)+A^{-10}.
\]
With $A=m^3$, $R=m^2$ and $K=m^5$ the leading term is $q\log A/m^7$, and $m^7=K^{7/5}$.  Since the
whole point set has $\Theta_\eta(K)$ points, this is the exponent $7/5$ of \cref{thm:main}.

\paragraph{Generalization.}
One natural question is whether this construction can be generalized to $\alpha > 2$ to obtain stronger lower bounds. In particular, {\em can we construct \yes{} and \no{} families for which distinguishing the two requires more than three local queries?} 
In \Cref{sec:tradeoff}, we sketch how the construction can be generalized to larger values of $\alpha$. However, while this modification increases the number of queries required to distinguish the two families, it comes at the cost of reducing the gap between their optimal Steiner tree costs. Our discussion outlines an argument that such a generalization leads to a query lower bound of $\widetilde{\Omega}(n^{(3\alpha+1)/(2\alpha+1)})$ for achieving an $\tfrac{\alpha}{\alpha-1}$-estimation.

\paragraph{Comparison with the earlier construction.}
We build on the framework of \cite[Section~7]{CKT23}: tight clusters of ordinary terminals, rare
special terminals hidden among them, and public blocks of nonterminal centers. Their special groups
are fixed across blocks. Each \no{} block uses a matching to connect special terminals to centers by unit-weight edges. The \yes{} instance replaces each group's matching by a star in a separately chosen hidden block.
The distinguishing witness is a center at distance $1$ from two special terminals.

Our improved exponent comes from two qualitatively new sources of hardness. First, our \no{} gadgets
themselves contain centers shared by pairs of special terminals. Each gadget occupies a public row
of centers, and we balance the two cases so that the complete view of any two group members has the
same distribution under the hidden center ordering. Thus the earlier shared-center witness is
inconclusive, and a third member must be probed against the row (\cref{lem:trade,lem:coupling}).
Second, we regroup the special terminals independently in each block, in contrast to the single
global grouping in \cite{CKT23}. Only one uniformly random hidden block carries the \yes{} gadgets.
Learning one block's grouping does not identify the grouping in another.  Our new uniform sparsity
bound (\cref{lem:sparsity}) ensures that few queries expose only a small fraction of the blocks, making it unlikely that any exposed block is the hidden one. Together with the indistinguishability of two-terminal views, this raises the query exponent from $6/5$ to $7/5$.

\subsection{Organization} \Cref{sec:prelim} fixes notation and records two standard facts.  \Cref{sec:constr} defines the two
distributions and \cref{sec:gap} computes their optimal costs exactly.  \Cref{sec:trade} proves the
local indistinguishability of the two families of centers and \cref{sec:coupling} turns it into a
coupling.  \Cref{sec:counting,sec:sparsity} bound the probability that the hidden block is exposed,
and \cref{sec:proof} assembles \cref{thm:main}.  \Cref{sec:tradeoff} generalizes the
construction and proves the trade-off between the exponent and the approximation
factor.

\section{Preliminaries}\label{sec:prelim}

\begin{definition}[Metric Steiner tree]\label{def:st}
The input is a triple $(V,T,w)$ in which $w$ is a metric on the finite set $V$ and $T\subseteq V$ is the set of terminals.
A feasible solution is a tree in the complete graph on $V$ whose vertex set contains $T$, and its
weight is the sum of the weights of its edges.  We write $\st(V,T,w)$ for the smallest weight of a
feasible solution.
\end{definition}

Requiring a tree rather than any connected subgraph loses nothing, since a connected subgraph
containing $T$ contains a tree containing $T$ of no larger weight.

An algorithm learns $w$ only by querying unordered pairs $\{u,v\}\subseteq V$ and receiving $w(u,v)$.
The set $V$, the set $T$, and the distribution the instance is drawn from are all known to it in
advance. Its query cost is the largest number of queries it makes, over all inputs and all random
choices. 

\paragraph{Transcripts.}For a deterministic algorithm $\mathcal A$ and an instance $w$, write $\trans(\mathcal A,w)$ for the
\emph{query--history transcript} of the run: the sequence
$\langle(\{u_1,v_1\},a_1),\dots,(\{u_q,v_q\},a_q)\rangle$ of queried pairs and their answers, where
$a_t=w(u_t,v_t)$. We extend the notation to distributions: if $\mathcal D$ is a distribution over instances, $\trans(\mathcal A,\mathcal D)$ denotes the distribution over the transcripts when the instance is drawn from $\mathcal D$.
We use two facts about the transcripts where both use only that $\mathcal A$ is deterministic:

\begin{enumerate}[label=\textup{(R\arabic*)},leftmargin=2.8em,itemsep=1pt]
  \item\label{it:R1} the $t$-th pair $\{u_t,v_t\}$ is determined by the first $t-1$ pairs and
        answers; and
  \item\label{it:R2} the transcript determines the output of $\mathcal A$.
\end{enumerate}

To show that $\mathcal A$ cannot tell two distributions $\mathcal D_1, \mathcal D_2$ apart, it suffices to make their transcripts agree with high probability, i.e., $\trans(\mathcal A,\mathcal D_1) = \trans(\mathcal A,\mathcal D_2)$ with high probability.

\medskip
For probability distributions $P,Q$ on a common finite set $X$ we write
$\dtv(P,Q)=\tfrac12\sum_{x\in X}|P(x)-Q(x)|$, and we use the identity
$\dtv(P,Q)=\max_{E\subseteq X}\bigl(P(E)-Q(E)\bigr)$, the maximum being over all subsets $E$ of the ground set.
We also use two standard facts.

\begin{proposition}[see, e.g., \cite{LeCam73}]\label{prop:lecam}
Let $P,Q$ be distributions on a finite set $\Omega$, let a fair coin decide which of the two the
observation is drawn from, and let a possibly randomized rule guess the answer from the observation.
The rule is correct with probability at most $\bigl(1+\dtv(P,Q)\bigr)/2$.
\end{proposition}

\begin{proposition}[Yao's min-max principle~\cite{Yao77}]\label{prop:yao}
Let $\mu$ be a distribution over inputs.  If every deterministic algorithm making at most $q$
queries is correct with probability less than $2/3$ on an input drawn from $\mu$, then every
randomized algorithm making at most $q$ queries is wrong with probability more than $1/3$ on at
least one input in the support of $\mu$.
\end{proposition}

\section{The two distributions}\label{sec:constr}

Throughout the rest of the paper we fix $\eta\in(0,1)$ and the constant integers
\begin{equation}\label{eq:d}
  d:=\Bigl\lceil\frac6\eta\Bigr\rceil-3,\qquad b:=(d-1)^2,
\end{equation}
so that $d>6/\eta-4$ and $d\ge4$.  We also fix the absolute constant $\theta:=\tfrac18$.  Let
$m\ge2$ be a growing integer and set 
\begin{equation}\label{eq:params}
  A:=m^3,\qquad R:=m^2,\qquad K:=AR=m^5,\qquad \lambda:=K^{-3},\qquad D:=\theta A .
\end{equation}

\paragraph{The points.}
The point set $V$ is fixed in advance and is known to the algorithm.  It consists of
\begin{itemize}[leftmargin=1.6em,itemsep=1pt]
  \item a set $T$ of $RA +dA = K+dA$ \emph{terminals}, and
  \item $R$ disjoint \emph{blocks} $C_1,\dots,C_R$ of non-terminals, called \emph{centers}, with
        $|C_i|=bA$ for each $i$. Each block is divided once and for all into $A$ labeled \emph{rows} $C_{i,1},\dots,C_{i,A}$ of $b$ centers each.
\end{itemize}
Thus
\begin{equation}\label{eq:ncount}
  n=(K+dA)+RbA=(b+1)K+dA=\Theta_\eta(K),
  \qquad |T|=K+dA=\Theta_\eta(n).
\end{equation}
Note that it is public which points are terminals, and that which block each center belongs to, and which row of that block it
lies in.
\paragraph{The hidden data.}
A \emph{configuration} $\omega$ consists of the following five objects, drawn independently and each
uniformly at random from its range.
\begin{enumerate}[label=\textup{(C\arabic*)},leftmargin=2.8em,itemsep=2pt]
  \item\label{it:C1} A set $S\subseteq T$ with $|S|=dA$.  Its members are the \emph{special}
        terminals. The other $K$ terminals are \emph{ordinary}.
  \item\label{it:C2} A partition of the ordinary terminals into labeled \emph{clusters}
        $G_1,\dots,G_A$, each of size $R$.
  \item\label{it:C3} For each block $i\in[R]$, a partition of $S$ into labeled \emph{groups}
        $S_{i,1},\dots,S_{i,A}$, each of size $d$.  The $R$ partitions are independent of one
        another. We use $U$ to refer to a generic group
        of $d$ special terminals that could replace any of $S_{i,j}$.
  \item\label{it:C4} For each block $i\in[R]$ and each row $j\in[A]$, a uniformly random ordering
        \[
          C_{i,j}(1),\ C_{i,j}(2),\ \dots,\ C_{i,j}(b)
        \]
        of the $b$ centers of that row. Note that the $RA$ orderings are independent of one another.
  \item\label{it:C5} An index $g\in[R]$, the \emph{hidden block index}.
\end{enumerate}
All configurations are equally likely.  Both distributions below draw $\omega$ this way. However the index
$g$ has no effect at all on the \no{} metric, which is exactly what makes it useful in
\cref{sec:coupling}.

\paragraph{Neighbor sets.}
Fix a block $i$ and a row $j$, and let $u_1,\cdots,u_d$ be any fixed ordering of the terminals in $S_{i,j}$. Now define the {\em neighbor set} of a center $C_{i,j}(r)$ denoted by $\nb_{i,j}(r)$ in each of \yes{} and \no{} instances as follows. 
For the \no{} instance set it as
\begin{equation}\label{eq:nbn}
  \mathsf \nb_{i,j}(r):=
  \begin{cases}
    \{u_a,u_c\}, & 1\le r\le\tbinom d2,\ \text{where $\{a,c\}$ is the $r$-th pair in a fixed order,}\\[2pt]
    \varnothing, & \tbinom d2<r\le b,
  \end{cases}
\end{equation}
and for the \yes{} instance set it as
\begin{equation}\label{eq:nby}
  \mathsf \nb_{i,j}(r):=
  \begin{cases}
    S_{i,j}, & r=1,\\[2pt]
    \{u_a\}, & r=1+(a-1)(d-2)+s\ \text{ with }a\in[d],\ s\in[d-2].
  \end{cases}
\end{equation}
Note that both lists have exactly $b$ entries.
We use $\mathsf Y$ and $\mathsf N$ to denote blocks in the \yes{} and \no{} instances, respectively, as defined above. Two consequences will be used repeatedly: in a \no{} block $\mathsf N$, every center has at most two neighbors, whereas in a \yes{} block $\mathsf Y$, each row contains one center adjacent to its entire group and $d(d-2)$ centers with one neighbor each. \Cref{fig:gadget} shows the two lists for $d=4$.

\begin{figure}[t]
\centering
\begin{tikzpicture}[x=1cm,y=1cm,
  term/.style={rectangle,draw=black!80,fill=black!72,minimum size=5.4pt,inner sep=0pt,rounded corners=0.5pt},
  ctr/.style ={circle,draw=black!60,fill=black!12,minimum size=5.4pt,inner sep=0pt},
  idle/.style={circle,draw=black!45,densely dashed,fill=white,minimum size=5.4pt,inner sep=0pt},
  hub/.style ={circle,draw=accent!60!black,fill=accent,minimum size=7pt,inner sep=0pt},
  lin/.style ={draw=black!42,line width=0.45pt},
  hlin/.style={draw=accent,line width=0.85pt},
  lbl/.style ={font=\scriptsize},
  slbl/.style={font=\tiny,text=black!60},
  panel/.style={font=\small\bfseries}]

\foreach \dx/\tag in {0/a, 7.6/b}{
  \begin{scope}[shift={(\dx,0)}]
    \draw[draw=black!16,rounded corners=8pt,fill=black!2]
      (-0.35,-0.75) rectangle (6.75,2.85);
    \foreach \k/\x in {1/1.2, 2/2.6, 3/4.0, 4/5.4}{
      \node[term] (t\k\tag) at (\x,2.15) {};
      \node[lbl] at (\x,2.52) {$u_\k$};
    }
    \foreach \k/\x in {0/0.4, 1/1.1, 2/1.8, 3/2.5, 4/3.2, 5/3.9, 6/4.6, 7/5.3, 8/6.0}{
      \coordinate (p\k\tag) at (\x,0.35);
    }
  \end{scope}
}

\foreach \k/\i/\j in {0/1/2, 1/1/3, 2/1/4, 3/2/3, 4/2/4, 5/3/4}{
  \draw[lin] (t\i a) -- (p\k a);
  \draw[lin] (t\j a) -- (p\k a);
}
\foreach \k in {0,...,5}{\node[ctr] at (p\k a) {};}
\foreach \k in {6,7,8}{\node[idle] at (p\k a) {};}
\node[panel] at (0.25,2.6) {(a)};
\node[slbl] at (3.2,-0.42) {six centers of degree two, three of degree zero};

\begin{scope}
  \coordinate (hubb) at (7.5, 0.40);
  \foreach \k/\i in {0/1, 1/1, 2/2, 3/2, 4/3, 5/3, 6/4, 7/4}{
    \draw[lin] (t\i b) -- (p\k b);
  }
  \foreach \i in {1,2,3,4}{\draw[hlin] (t\i b) -- (hubb);}
  \foreach \k in {0,...,7}{\node[ctr] at (p\k b) {};}
  \node[hub] at (hubb) {};
\end{scope}
\node[panel] at (7.85,2.6) {(b)};
\node[slbl] at (10.8,-0.42) {eight centers of degree one, one of degree four};

\node[font=\scriptsize,text=black!70] at (7.2,-1.15)
  {a line means distance $1$; every pair not joined by a line is at distance $2$};
\end{tikzpicture}
\caption{The $b$ centers of one group's row, for a group $U=\{u_1,u_2,u_3,u_4\}$ of $d=4$ special terminals, under
$\mathsf N$ in (a) and $\mathsf Y$ in (b).  In both pictures the group's row has
$b=(d-1)^2=9$ centers, each terminal has exactly $d-1=3$ centers at distance $1$, and there are
exactly $d(d-1)=12$ terminal--center pairs at distance $1$.  By \cref{lem:trade} the agreement goes
further: from the viewpoint of any one or two terminals of the group, the two pictures are the
same.  Only three terminals of one group can tell them apart, because in (b) one center is at
distance $1$ from three of them, which never happens in (a).}\label{fig:gadget}
\end{figure}
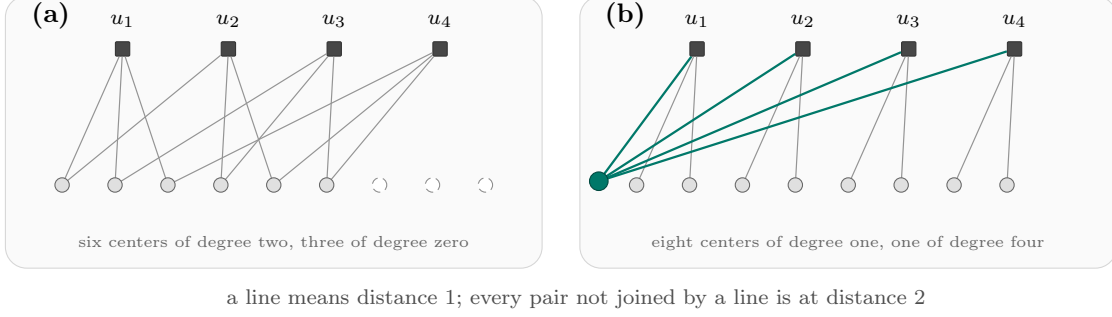

\paragraph{The two metrics.}
Given a configuration $\omega$, the \no{} metric $w_{\mathsf n}$ uses $\mathsf N$ in every
block, and the \yes{} metric $w_{\mathsf y}$ uses $\mathsf Y$ in block $g$ and $\mathsf N$ in every other block. In both cases, for distinct points $x\ne y$,
\begin{equation}\label{eq:w}
  w(x,y)=
  \begin{cases}
    \lambda, & x,y\text{ are ordinary terminals of a common cluster},\\
    1, & x\in C_{i,j} \text{ and } y\in\nb_{i,j}(x)\text{, or the same with $x,y$ swapped},\\
    2, & \text{otherwise},
  \end{cases}
\end{equation}
and $w(x,x)=0$.  

\paragraph{Distributions.} Let $\Dno$ and $\Dyes$ be the resulting distributions on instances, for $\omega$
chosen uniformly at random.  Finally the input distribution $\mu$ of our algorithm will be $\Dyes$ with probability $1/2$ and $\Dno$ with probability $1/2$, i.e., $\mu := \tfrac12\Dno+\tfrac12\Dyes$.

\begin{lemma}\label{lem:metric}
Every instance in the support of $\Dno$ or of $\Dyes$ is a metric.
\end{lemma}

\begin{proof}
The function $w$ is symmetric, vanishes exactly on the diagonal, and is positive elsewhere, so only
the triangle inequality needs checking.  Since every distance is at most $2$, a violation
$w(x,z)>w(x,y)+w(y,z)$ would need $w(x,y)+w(y,z)<2$, so the two distances on the right are either
$\lambda$ and $\lambda$, or $\lambda$ and $1$ in some order.

If $w(x,y)=w(y,z)=\lambda$, then $x,y$ lie in a common cluster and so do $y,z$. Further the clusters are
disjoint, so $x$ and $z$ lie in a common cluster too and $w(x,z)=\lambda\le2\lambda$.

If one of the two distances is $\lambda$ and the other is $1$, say $w(x,y)=\lambda$ and
$w(y,z)=1$, then $y$ is an ordinary terminal by the first line of \eqref{eq:w}, while the second
line of \eqref{eq:w} forces $y$ to be either a center or a special terminal.  This is impossible,
so no such pair of distances occurs.
\end{proof}

\subsection{Comparison of the Optimal Costs of the Two Instances}\label{sec:gap}

We compute $\st(V,T,w)$ exactly in both worlds.  Note first that in either \yes{} or \no{} instance, the optimal solution may be
assumed to contain, inside each cluster, a spanning tree of $\lambda$-edges: all $R$ members of a
cluster are at distance $2$ from every point outside the cluster, so any solution may be rerouted so
that only one member of the cluster meets the rest of the tree, and the other $R-1$ members are connected to it at cost $\lambda$ each.  This costs $A(R-1)\lambda=\lambda(K-A)$ in total and leaves $\kappa:=A+dA=(d+1)A$ \emph{terminal pieces}: the $A$ contracted clusters and the $dA$ special terminals.

For a set $X$ of centers (or nonterminals), let $H_X$ be the bipartite graph whose vertices are the $\kappa$ terminal
pieces together with the centers in $X$, with an edge whenever the distance is $1$.  Only special
terminals have such edges, since ordinary terminals are at distance $2$ from every center.  Let
$r(X)$ denote the largest number of edges in a cycle-free subgraph of $H_X$; equivalently,
$r(X)=(\kappa+|X|)-(\text{number of connected components of }H_X)$.

\begin{lemma}
\label{lem:profit}
In either instance, $\st(V,T,w)=\lambda(K-A)+2(\kappa-1)-\max_X\bigl(r(X)-2|X|\bigr)$,
the maximum being over all sets $X$ of centers.  Moreover
\[
  \max_X\bigl(r(X)-2|X|\bigr)=
  \begin{cases}
    0,& \text{in \no{} instances},\\
    (d-2)A,& \text{in \yes{} instances}.
  \end{cases}
\]
\end{lemma}

\begin{proof}
\emph{Upper bound.}  Fix $X$ and take a cycle-free subgraph of $H_X$ with $r(X)$ edges.  Added to
the $\lambda$-spanning trees of the clusters, these $r(X)$ edges of weight $1$ leave
$\kappa+|X|-r(X)$ connected pieces among the terminal pieces and the centers of $X$.  Any two
distinct pieces contain a pair of points at distance exactly $2$: such a pair cannot be at distance
$\lambda$, because each cluster lies inside a single piece, and it cannot be at distance $1$, because
an edge of $H_X$ between two distinct pieces could have been added without creating a cycle,
contradicting the choice of $r(X)$.  So the pieces can be connected by $\kappa+|X|-r(X)-1$ edges of
weight $2$. So the total weight of the resulting tree is
\[
  \lambda(K-A)+r(X)+2\bigl(\kappa+|X|-r(X)-1\bigr)
  =\lambda(K-A)+2(\kappa-1)-\bigl(r(X)-2|X|\bigr).
\]

\emph{Lower bound.}  Let $\Upsilon$ be any feasible tree and let $X$ be the set of centers it uses,
so that $\Upsilon$ has $K+dA+|X|$ vertices and $K+dA+|X|-1$ edges.  Split its edges by weight into
$F_\lambda$, $F_1$ and $F_2$.  Since the $\lambda$-edges of $\Upsilon$ form a cycle-free subgraph of
the disjoint union of the $A$ clusters, as we discussed earlier, $|F_\lambda|\le A(R-1)=K-A$.  Contracting the connected
pieces of $(V(\Upsilon),F_\lambda)$ turns $\Upsilon$ into a tree, so
$|F_1|+|F_2|=K+dA+|X|-1-|F_\lambda|$ and the weight of $\Upsilon$ equals
\begin{align*}
    \lambda|F_\lambda|+|F_1|+2|F_2| 
    &= (\lambda-2)|F_\lambda| + 2\bigl(K+dA+|X|-1\bigr)-|F_1| &&\rhd \lambda<2\\
    &\ge (\lambda-2) (K-A) + 2\bigl(K+dA+|X|-1\bigr)-|F_1| \\
    &= \lambda(K-A) + 2\bigl(A+dA+|X|-1\bigr)-|F_1| \\
    &= \lambda(K-A)+2(\kappa-1)+2|X|-|F_1|. &&\rhd \kappa = (d+1)A
\end{align*}
Finally, the edges of $F_1$ join special terminals to centers of $X$, so they are edges of $H_X$,
and they form no cycle because $\Upsilon$ is a tree and contracting the clusters does not affect
them.  Hence $|F_1|\le r(X)$ and the weight of $\Upsilon$ is at least
$\lambda(K-A)+2(\kappa-1)-\bigl(r(X)-2|X|\bigr)$.  This proves the identity.

\medskip
\emph{The saving in \no{} instances.}
Group the vertices of $H_X$ by connected component, and let $t$ and $x$ denote the
numbers of terminal pieces and centers in a component. Then
\begin{equation}\label{eq:decomp}
  r(X)-2|X|
  =\sum_{\text{components}}\bigl((t+x-1)-2x\bigr)
  =\sum_{\text{components}}(t-x-1).
\end{equation}
Consider one component. Since it is connected, it has $t+x-1$ edges.
On the other hand, every edge is incident to exactly one center, so its number
of edges is the sum of the degrees of its $x$ centers. In a \no{} block
every center has degree at most $2$, and hence the component has at most $2x$ edges. Therefore $t+x-1\le 2x$, or equivalently $t-x-1\le0$. Thus every term in \eqref{eq:decomp} is
nonpositive, so $r(X)-2|X|\le0$ for every $X$. Equality is attained for
$X=\varnothing$, when every component consists of a single terminal piece (i.e., the constructed tree uses no nonterminal vertex.).

\medskip
\emph{The saving in \yes{} instances.} Consider a component of $H_X$ containing $t$ terminal
pieces, $h$ centers of degree $d$ and $\ell$ centers of degree at most $2$. By the two counts
just made similar to the above we have $t+h+\ell-1\le dh+2\ell$, so $t-(h+\ell)-1=(t+h+\ell-1)-2(h+\ell)\le(d-2)h$. Centers of degree $d$ occur only in block $g$, one per group, so summing over components bounds
\eqref{eq:decomp} by $(d-2)A$.  This is attained by taking $X$ to be exactly those $A$ centers: the components of $H_X$ are then the $A$ stars, each with $t=d$ and $x=1$.
\end{proof}
By \cref{lem:profit}, the costs of the optimal trees in these two instances are the following values.
\begin{align}
  \Theta_{\mathsf n}&:=\lambda(K-A)+2\bigl((d+1)A-1\bigr),\label{eq:optno}\\
  \Theta_{\mathsf y}&:=\lambda(K-A)+2\bigl((d+1)A-1\bigr)-(d-2)A=\lambda(K-A)+(d+4)A-2,\label{eq:optyes}
\end{align}
the same for every configuration.  Since $\lambda(K-A)\le K^{-2}\to0$ and $A\to\infty$, the ratio
$\Theta_{\mathsf n}/\Theta_{\mathsf y}$ tends to $2(d+1)/(d+4)=2-6/(d+4)$, which by \eqref{eq:d} is
larger than $2-\eta$.  Hence there is an $m_\eta$ with
\begin{equation}\label{eq:gap}
  \Theta_{\mathsf n}>(2-\eta)\,\Theta_{\mathsf y}\qquad\text{for all }m\ge m_\eta.
\end{equation}

\subsection{One or two terminals see nothing}\label{sec:trade}
Fix a block $i$ and a row $j$, and let $U:=S_{i,j}$, and for any $1\leq r\leq b$, write $\nb^{\mathsf N}_{i,j}(r)$ and
$\nb^{\mathsf Y}_{i,j}(r)$ for the neighbor set of $C_{i,j}(r)$ given by \eqref{eq:nbn} and by \eqref{eq:nby}. For $B,H\subseteq U$ call $B\cap H$ the \emph{view of $B$ from $H$}.

\begin{lemma}
\label{lem:trade}
Let $H\subseteq U$ with $|H|\le2$.  There is a permutation $\phi_H$ of $[b]$ with
\[
  \nb^{\mathsf N}_{i,j}(r)\cap H=\nb^{\mathsf Y}_{i,j}\bigl(\phi_H(r)\bigr)\cap H
  \qquad\text{for every }r\in[b].
\]
\end{lemma}
The lemma will be later used in the following form.  Let $H$ be the set of members of $S_{i,j}$ that the
algorithm has probed against centers of $C_{i,j}$, and suppose $|H|\le2$.  Then reordering that row
carries its \no{} neighbor sets onto its \yes{} ones without changing the answer to any pair already asked, and by \ref{it:C4} the ordering of a row is hidden and uniform, so the reordering costs nothing: on that row the two instances are indistinguishable.  Note that the terminals outside $S_{i,j}$ do not matter as they all lie at distance $2$ from every center of $C_{i,j}$ in both instances.  
\begin{proof}
It suffices to check that, for each possible view, the two blocks have the same number of centers
whose neighbor set has that view, and matching them within each view class defines $\phi_H$.

Let $H=\{u,v\}$ with $u\ne v$.  In $\mathsf N$ exactly one neighbor set is $\{u,v\}$; exactly
$d-2$ are pairs $\{u,x\}$ with $x\notin H$, and likewise $d-2$ are pairs $\{v,x\}$; and the
remaining ones, namely the $\binom{d-2}2$ pairs avoiding $H$ together with the $\tfrac{(d-1)(d-2)}2$ empty ones, have empty view. Their count is
\[
  \binom{d-2}2+\frac{(d-1)(d-2)}2=\frac{(d-2)\bigl((d-3)+(d-1)\bigr)}2=(d-2)^2 .
\]
In $\mathsf Y$ exactly one neighbor set is $U$, whose view is $\{u,v\}$; the $d-2$ equal to $\{u\}$ have view $\{u\}$, and likewise for $v$; and the $(d-2)^2$ equal to $\{x\}$ for $x\notin H$ have empty view. The two counts agree, and they sum to $\bigl(1+(d-2)\bigr)^2=(d-1)^2=b$ as they should.
\[
\begin{array}{c|cccc}
\text{view}&\{u,v\}&\{u\}&\{v\}&\varnothing\\\hline
\text{number of centers}&1&d-2&d-2&(d-2)^2
\end{array}
\]

Let $H=\{u\}$.  In $\mathsf N$ the neighbor sets with view $\{u\}$ are the $d-1$ pairs containing
$u$. In $\mathsf Y$ they are $U$ together with the $d-2$ copies of $\{u\}$, again $d-1$ in all.  The
remaining $b-(d-1)$ have empty view in either block.  For $H=\varnothing$ every view is empty and
both blocks have $b$ centers per row.
\end{proof}

The count in \cref{lem:trade} is tight: for $|H|=3$ the neighbor set $U$ has a view of size $3$ in
$\mathsf Y$, while in $\mathsf N$ every view has size at most $2$.  This is the only way the
two instances can be told apart locally, and the next section makes that precise.

\section{The Coupling}\label{sec:coupling}
Given a configuration, its \yes{} and \no{} instances differ only inside the hidden block $g$, and by \cref{lem:trade} that difference is indistinguishable to any one or two members of a group. This section turns that local statement into one about the whole run: unless three members of some single group of block $g$ have been probed against their corresponding centers set, the transcript of a \no{} run could equally well have come from a \yes{} or instance. So the two distributions can be told apart only through that one event, and \cref{sec:counting,sec:sparsity} bound its probability.

\medskip 
Fix a deterministic algorithm $\mathcal A$ that makes $q$ queries and run it on a \no{} instance.
 
\begin{definition}\label{def:exposed}
For a block $i$ and a row $j$, let
\begin{equation}\label{eq:touch}
  \touch_i(j):=\bigl\{u\in S_{i,j}:\ \mathcal A\text{ queried the pair }\{u,c\}
  \text{ for some center }c\in C_{i,j}\bigr\}
\end{equation}
be the subset of the special terminals $S_{i,j}$ that were probed against the centers in $C_{i,j}$. The pair $(i,j)$ is \emph{exposed} if $|\touch_i(j)|\ge3$, and the block $i$ is \emph{exposed} if $(i,j)$ is exposed for some $j$.
\end{definition}

\begin{lemma}
\label{lem:coupling}
For every deterministic algorithm $\mathcal A$ making $q$ queries,
\[
  \dtv\bigl(\trans(\mathcal A,\Dno),\,\trans(\mathcal A,\Dyes)\bigr)
  \ \le\ \Prob_{\Dno}\bigl[\,\text{block $g$ is exposed}\,\bigr].
\]
\end{lemma}
 
\begin{proof}
Let $\omega$ be a configuration for which block $g$ is not exposed, in the sense of
\cref{def:exposed}, in the \no{} run.  Define a new
configuration $\omega'$ by keeping \ref{it:C1}, \ref{it:C2}, \ref{it:C3}, \ref{it:C5} and the ordering of every row outside block $g$ unchanged, and reordering the rows of block $g$ by setting
\[
    C'_{g,j}\bigl(\phi_{\touch_g(j)}(r)\bigr):=C_{g,j}(r) \qquad\text{for all }j\in[A],\ r\in[b],
\]
where $\phi_H$ is the map of \cref{lem:trade}, applied to the group $S_{g,j}$, i.e., $H=\touch_g(j)$. Note that this is legitimate
because $|\touch_g(j)|\le2$ for every $j$.  Since each $\phi_{\touch_g(j)}$ is a permutation of
$[b]$, each $C'_{g,j}$ is again an ordering of the row $C_{g,j}$, so $\omega'$ is a configuration.
Write $\Psi(\omega):=\omega'$.
 
\smallskip
\emph{The \yes{} instance of $\omega'$ gives the same transcript as the \no{} instance of $\omega$.}
We show by induction on $t$ that the first $t$ pairs and answers agree.  Given that the first $t-1$
agree, property \ref{it:R1} says the $t$-th pair is the same in both runs, so it suffices to check
that this pair receives the same answer. Distances between two terminals remain unchanged by
\ref{it:C1} and \ref{it:C2}, and two centers are always at distance $2$.
It remains to look at a pair $\{u,c\}$ with $u$ a terminal and $c$ is a center in $C_g$. 
The answer to such a pair is $1$ when $u$ lies in the neighbor set of $c$ and $2$ otherwise, so we need to ensure that $u$ lies in the neighbor set of $c$ in one instance exactly when it does in the other. 
Let $j$ be the index of the row containing $c$, and let $r$ be the index with $c=C_{g,j}(r)$ in $\omega$, so that $c=C'_{g,j}\bigl(\phi_{\touch_g(j)}(r)\bigr)$ in $\omega'$. Since the rows are fixed once and for all, we have that  $j$ and $S_{g,j}$ are the same in both. 
The answer to $\{u,c\}$ is thus decided by $\nb^{\mathsf N}_{g,j}(r)$ in the \no{} run and by $\nb^{\mathsf Y}_{g,j}(\phi_H(r))$ in the \yes{} run. Both are subsets of $S_{g,j}$, so if $u\notin S_{g,j}$ the answer is $2$ in both. If $u\in S_{g,j}$, then $u\in H$.
Block $g$ is not exposed, so $|H|\le2$ and \cref{lem:trade} applies: the two sets just named have the same intersection with $H$. Since $u$ lies in $H$, it belongs to both of them or to neither, and the answer is again the same in
both runs.
 
\smallskip
\emph{$\Psi$ is one-to-one.}  Suppose $\omega'=\Psi(\omega)$ is given.  The transcript of the \yes{} run
on $\omega'$ equals that of the \no{} run on $\omega$, by the previous paragraph, and the rows
are public. So, by \eqref{eq:touch} the sets $\touch_g(j)$ are determined by $\omega'$.  Hence so are the maps $\phi_{\touch_g(j)}$, and the orderings $C_{g,j}$ can be recovered from the
$C'_{g,j}$ by inverting them.  All other parts of $\omega$ are visible in $\omega'$ already.
 
Now fix any set $E$ of transcripts.  Write $\Omega$ for the finite set of all configurations and,
for a configuration $\omega$, write $w_{\mathsf n}(\omega)$ and $w_{\mathsf y}(\omega)$ for the two
metrics it determines; both $\Dno$ and $\Dyes$ draw $\omega$ uniformly from $\Omega$, so the two
worlds differ only in which of the two metrics the algorithm is run on. Define
\[
  \mathcal E_{\mathsf n}:=\bigl\{\omega\in\Omega:\ \trans(\mathcal A,w_{\mathsf n}(\omega))\in E\bigr\},
  \qquad
  \mathcal E_{\mathsf y}:=\bigl\{\omega\in\Omega:\ \trans(\mathcal A,w_{\mathsf y}(\omega))\in E\bigr\},
\]
and let $\mathcal G$ be the set of configurations whose block $g$ is not exposed in the \no{} run,
so that $\Psi$ is defined on $\mathcal G$.  Since $\Psi$ injects $\mathcal E_{\mathsf n}\cap\mathcal G$
into $\mathcal E_{\mathsf y}$, we have
$|\mathcal E_{\mathsf n}\cap\mathcal G|\le|\mathcal E_{\mathsf y}|$, and dividing by $|\Omega|$
turns this into $\Prob_{\Dno}[\,E,\ g\text{ not exposed}\,]\le\Prob_{\Dyes}[E]$.
Thus
\[
  \Prob_{\Dno}[E]
  \le\Prob_{\Dno}[\,g\text{ exposed}\,]+\Prob_{\Dno}[\,E,\ g\text{ not exposed}\,]
  \le\Prob_{\Dno}[\,g\text{ exposed}\,]+\Prob_{\Dyes}[E].
\]
So, $\dtv(\trans(\mathcal A,\Dno),\trans(\mathcal A,\Dyes))=\max_E\bigl(\Prob_{\Dno}[E]-\Prob_{\Dyes}[E]\bigr) \le\Prob_{\Dno}[\,g\text{ exposed}\,]$. 
\end{proof}
 
It remains to show that an algorithm making few queries exposes few blocks.  The difficulty is that the algorithm chooses which terminals to investigate as it goes, using everything it has seen.

\section{What few queries can reveal}\label{sec:counting}

We showed to bound $\dtv\bigl(\trans(\mathcal A,\Dno),\,\trans(\mathcal A,\Dyes)\bigr)$ it suffices to bound $\Prob_{\Dno}[\,g\text{ exposed}\,]$. In this section (and~\Cref{sec:sparsity}) we assume the algorithm is run on a \no{} instance.
We assume no pair is queried twice, which only helps the algorithm.  We write $\mathcal F_t$ for everything in the transcript
before the $t$-th query, that is, for the prefix
$\langle(\{u_1,v_1\},a_1),\dots,(\{u_{t-1},v_{t-1}\},a_{t-1})\rangle$. Note that by \ref{it:R1} the $t$-th pair is determined by $\mathcal F_t$.

\subsection{A counting principle}

To show that few configurations are {\em bad}, we find at least $F$ ways to change each bad configuration into one that is still consistent with what the algorithm has seen, and we check that no bad configuration is reached in more than $B$ ways. Counting these changes twice then bounds the fraction of bad configurations by $B/F$, which is the lemma below.
\begin{lemma}
\label{lem:exchange}
Let $\Omega$ be a finite set, let $\mathcal E\subseteq\Omega$, and let $\mathcal M\subseteq\mathcal E\times\Omega$ be a set
of ordered pairs, called \emph{moves}.  Suppose every $\xi\in\mathcal E$ appears as the first entry of at
least $F$ moves and every $\zeta\in\Omega$ appears as the second entry of at most $B$ moves.  Then
$|\mathcal E|\le(B/F)\,|\Omega|$.
\end{lemma}
 
\begin{proof}
Counting the moves in two ways, $F|\mathcal E|\le|\mathcal M|\le B|\Omega|$.
\end{proof}
 
We apply this as follows.  Condition on the transcript before time $t$ and let $\Omega_t$ be the set of
configurations consistent with it, i.e., the set of configurations under which the first $t-1$ queries of $\mathcal A$ receive exactly the answers listed in $\mathcal F_t$.  
Since all configurations are equally likely and the transcript is a function of the configuration, the conditional distribution of the configuration given $\mathcal F_t$ is uniform on $\Omega_t$.
 
Let $\mathcal E\subseteq\Omega_t$ be an event whose conditional probability we want to bound.  We use \cref{lem:exchange} with $\Omega_t$ in the role of $\Omega$: we exhibit a set $\mathcal M\subseteq\mathcal E\times\Omega_t$ of moves such that every $\xi\in\mathcal E$ is the first entry of at least $F$ of them and every $\zeta\in\Omega_t$ is the second entry of at most $B$ of them.  The lemma then gives $|\mathcal E|\le(B/F)|\Omega_t|$, and since the conditional law is uniform on $\Omega_t$,
\begin{equation}\label{eq:exchapply}
  \Prob[\mathcal E\mid\mathcal F_t]=\frac{|\mathcal E|}{|\Omega_t|}\le\frac BF .
\end{equation}
Note that both $F$ and $B$ may depend on $t$ and on the transcript, which create no issue since \eqref{eq:exchapply} is a bound on a conditional probability.
 
\medskip
The moves are of one kind throughout.  Given a configuration $\xi$ and two terminal labels $u$
and $x$, \emph{transposing $u$ and $x$} means putting each in the other's place in \ref{it:C1},
\ref{it:C2} and \ref{it:C3} at once, and leaving \ref{it:C4} and \ref{it:C5} alone: the two labels
exchange their roles as special or ordinary, their positions in the clusters, and their positions in
every one of the $R$ group partitions, so that the subsets attached to the centers have $u$ replaced by $x$
and vice versa.  The result is the second entry of the move.  Two things then have to be checked in
each application, and they are what the exclusion rules in the proofs below are for.  A move must
\emph{land in $\Omega_t$}: the transposition must leave every one of the first $t-1$ answers as it
was.
Moreover, one must bound $B$, by bounding the number of moves that could have produced a given target $\zeta$.

\subsection{Bookkeeping}

Call an ordinary terminal \emph{marked} once some query involving it has returned $\lambda$.  Let
$\rho$ be the first time at which $\theta K$ ordinary terminals are marked, and set $\rho:=q+1$ if
this never happens. Define the event $\mathcal R:=\{\rho\le q\}$.  All quantities below refer to the first $\rho-1$ queries; conditioned on $\mathcal R^{\mathsf c}$ (i.e., $q>\rho$), all quantities hold for the first $q$ queries.

For a terminal $u$ let $\deg_t(u)$ be the number of queries involving $u$ before time $t$.  For a
cluster (of ordinary terminals) $G$ let $e_t(G)$ be the number of terminals outside $G$ that have been queried against a member of $G$ before time $t$; since a query has two ends,
\begin{equation}\label{eq:ebudget} 
\sum_{a\in[A]}e_{q+1}(G_a)\le2q .
\end{equation}
For a special terminal $s$ let
\[
  b_t(s):=\#\bigl\{x\text{ ordinary}:\ \{x,c\}\text{ was queried before time }t
  \text{ for some center }c\text{ at distance }1\text{ from }s\bigr\}
\]
count the ordinary terminals that have been tried against a center locating at distance $1$ from $s$.
In a \no{} block every center has at most two terminals at distance $1$, so each query
contributes to $b_{q+1}(s)$ for at most two terminals $s$ and
\begin{equation}\label{eq:bbudget}
  \sum_{s\in S}b_{q+1}(s)\le2q .
\end{equation}

\subsection{Clusters are hard to find}

\begin{lemma}\label{lem:cluster}
$\Prob[\mathcal R]=O\!\left(\dfrac q{AK}\right).$
\end{lemma}

\begin{proof}
Let $M$ be the number of terminals that get marked at some time $t\le\rho$.  On $\mathcal R$ we have $M\ge\theta K$, so by Markov's inequality,  $\Prob[\mathcal R]\le\Ex[M]/(\theta K)$ and to show the lemma it is enough to prove $\Ex[M]=O(q/A)$.  Classify a mark, occurring at time $t\le\rho$ at a terminal $u$ lying in cluster
$G$, by the first of the following that applies.
 
\emph{(i) $\deg_t(u)\ge D$.}  At most $2q/D=O(q/A)$ terminals are ever queried $D$ times, and each is
marked at most once.
 
\emph{(ii) $e_t(G)\ge\theta K$.}  By \eqref{eq:ebudget} at most $2q/(\theta K)$ clusters ever reach
$e_{q+1}(G)\ge\theta K$, and each has $R$ members, so this accounts for at most
$2qR/(\theta K)=O(q/A)$ marks.
 
\emph{(iii) $\deg_t(u)<D$ and $e_t(G)<\theta K$.}  Fix $t$ and one of the two ends $u$ of the $t$-th
pair, and let $v$ be the other end.  Let $\mathcal E$ be the set of configurations in $\Omega_t$ for which
$u$ is unmarked, $u$ and $v$ lie in a common cluster $G$ with $e_t(G)<\theta K$, and
$\deg_t(u)<D$. Now on $\mathcal E$ the $t$-th answer is $\lambda$ and marks $u$.  As moves, take the
transposition of $u$ with an ordinary terminal $x$ outside $G$, excluding $x$ if
\begin{enumerate}[label=(\alph*),leftmargin=2.2em,itemsep=1pt]
  \item $x$ is marked;
  \item $x$ has been queried against a member of $G$; or
  \item the cluster of $x$ contains a terminal that has been queried against $u$.
\end{enumerate}

Rule (a) rules out fewer than $\theta K$ labels because $t\le\rho$, rule (b) fewer than $\theta K$
because $e_t(G)<\theta K$, and rule (c) at most $DR=\theta K$ because $\deg_t(u)<D$ and each cluster has $R$ members; the requirement $x\notin G$ rules out $R$ more. As $\theta=1/8$, at least $K-3\theta K-R\ge\tfrac12K$ choices of $x$ remain, since $K-3\theta K=\tfrac58K$ and
$R/K=1/A=m^{-3}\le\tfrac18$ for every $m\ge2$, so $F=\Omega(K)$ (as in \cref{lem:exchange}).
 
Each such move stays in $\Omega_t$.  Queries involving neither $u$ nor $x$ are unaffected, since
only the positions of $u$ and $x$ in \ref{it:C2} change and both remain ordinary, so no
terminal--center answer changes either.  Neither $u$ nor $x$ has received the answer $\lambda$, so
no earlier answer is destroyed. Further, no earlier answer becomes $\lambda$, since a new $\lambda$ at
$u$ would need an earlier query between $u$ and the cluster of $x$, excluded by (c), while a new
$\lambda$ at $x$ would need an earlier query between $x$ and $G$, excluded by (b).
 
Finally, fix a configuration $\zeta$ reached by such a move.  In $\zeta$ the terminal $x$ lies in the
cluster of $v$, and the move is determined by $x$, so at most $R$ moves reach $\zeta$ and $B\le R$.
By \cref{lem:exchange}, $\Prob[\mathcal E\mid\mathcal F_t]=O(R/K)=O(1/A)$.  Summing over the $2q$ ends of the $q$ queries bounds the expected number of marks of type (iii) by $O(q/A)$.

\smallskip
Summing over all three cases, $\Ex[M] = O(q/A)$. 
\end{proof}

\subsection{Special terminals are hard to recognize}

\begin{lemma}
\label{lem:special}
Let $t\le\rho$ and suppose that, according to $\mathcal F_t$, the terminal $u$ satisfies
$\deg_t(u)<D$ and has never received the answer $1$.  Then, for every block $i$, every row $j$ and
every center $c$, 
\begin{align}
  \Prob\bigl[u\in S,\ b_t(u)<\theta K\ \big|\ \mathcal F_t\bigr]&=\text{$O(d/R)$},\label{eq:post1}\\
  \Prob\bigl[u\in S,\ b_t(u)<\theta K,\ w(u,c)=1\ \big|\ \mathcal F_t\bigr]&=O(1/K),\label{eq:post2}\\
  \Prob\bigl[u\in S,\ b_t(u)<\theta K,\ u\in S_{i,j}\ \big|\ \mathcal F_t\bigr]&=\text{$O(d/K)$}.
    \label{eq:post3}
\end{align}
\end{lemma}

\begin{proof}
One argument proves all three bounds.  The moves are the same in each case, so the lower bound
$F=\Omega(K)$ on the number of moves available at a configuration is shown once and for all. What separates the three bounds from each other is only the upper bound $B$ on the number of moves that can reach a given target.

Fix which of the three bounds is being proved. Let $\mathcal E$ be the set of configurations in $\Omega_t$ satisfying the event in question. As moves,
consider swapping the role of $u$ with an ordinary terminal $x$: i.e., put $u$ in the position of $x$
in \ref{it:C2} and put $x$ in place of $u$ in \ref{it:C1} and in every one of the $R$ group
partitions \ref{it:C3}, so that the subsets attached to the centers have $u$ replaced by $x$. Exclude $x$
if
\begin{enumerate}[label=(\alph*),leftmargin=2.2em,itemsep=1pt]
  \item $x$ is marked;
  \item the cluster of $x$ contains a terminal that has been queried against $u$; or
  \item $x$ is one of the terminals counted by $b_t(u)$.
\end{enumerate}

Rule (a) rules out fewer than $\theta K$ labels since $t\le\rho$, rule (b) at most $DR=\theta K$
since $\deg_t(u)<D$, and rule (c) fewer than $\theta K$ on $\mathcal E$. Hence at least
$K-3\theta K=\tfrac58K$ choices remain and $F=\Omega(K)$.

Each move stays in $\Omega_t$.  Queries involving neither $u$ nor $x$ are unaffected.  A query
between $u$ and another terminal was answered $2$ before the move, since $u$ was special, and could
become $\lambda$ only if $u$'s new cluster had already been queried against $u$, which (b) excludes.
A query between $x$ and another terminal was answered $2$ before the move, because $x$ was unmarked
by (a), and is answered $2$ after it, because $x$ is now special.  A query between $u$ and a center
was answered $2$ before the move by hypothesis, and is answered $2$ after it because $u$ is now
ordinary.  A query between $x$ and a center was answered $2$ before the move and could become $1$
only if that center were at distance $1$ from $u$ before the move, which (c) excludes.

\smallskip
It remains to count, for a configuration $\zeta$ reached by a move, how many moves reach it.  A move
is determined by the label $x$, and in $\zeta$ that label is special, so at most $|S|=dA$ moves reach
$\zeta$ and \eqref{eq:post1} follows from \cref{lem:exchange}:
$\Prob[\mathcal E \mid\mathcal F_t] = O(dA/K) = O(d/R)$, as $K=AR$. For \eqref{eq:post2}, the move turns $w(u,c)=1$ into
$w(x,c)=1$, and in the \no{} world at most two terminals are at distance $1$ from $c$, so at most two
moves reach $\zeta$ and $\Prob[\mathcal E \mid\mathcal F_t]=O(1/K)$.  For \eqref{eq:post3}, the move places
$x$ in the group $S_{i,j}$, which has $d$ members, so at most $d$ moves reach $\zeta$ and
$\Prob[\mathcal E \mid\mathcal F_t]=O(d/K)$.
\end{proof}

\subsection{Few noticed terminals, few lucky queries}

The three bounds in \cref{lem:special} are proved under the same hypotheses for a terminal:
no earlier query to it was answered $1$, it has received fewer than $D$ earlier queries, and
$b_t<\theta K$.  Terminals that fail at least one of these hypotheses are precisely the ones
on which the algorithm has already obtained some leverage.  We collect them in the following definition. 
\begin{definition}\label{def:noticed}
A special terminal is \emph{noticed} if, at some time $t\le\rho$, it receives the answer $1$, or its
query count reaches $D$, or $b_t$ reaches $\theta K$.  We write $W$ for the set of noticed special
terminals.
\end{definition}
 
Few terminals are noticed in $q$ queries.

\begin{lemma}
\label{lem:W}
$\Ex\bigl[|W|\bigr]= O\!\left(\dfrac{dq}K\right)$.
\end{lemma}

\begin{proof}
Charge each noticed terminal to the first of the three events that makes it noticed.

If the answer $1$ comes first, the terminal had no earlier answer $1$, fewer than $D$ earlier
queries, and $b_t<\theta K$ at the time.  A pair receiving the answer $1$ consists of a terminal and
a center, so \eqref{eq:post2} applies with $c$ the center end of the $t$-th pair and $u$ its terminal
end, and bounds the conditional probability by $O(1/K)$.  Summing over the at most $q$ queries gives
$O(q/K)$. 

If the count $D$ comes first, then at the time $u$ has had no answer $1$ and $b_t(u)<\theta K$, so
\eqref{eq:post1} bounds the conditional probability that $u$ is special by $O(d/R)$.  
To turn this into a bound on the expected number of terminals charged here, let $\mathcal{T}$ be the set
of times $t\le\rho$ at which the $t$-th query brings some terminal's count to $D$, and for $t\in \mathcal{T}$
let $u_t$ be that terminal.  Whether $t\in \mathcal{T}$, and the identity of $u_t$, are determined by
$\mathcal F_t$: query counts are read off the transcript, and by \ref{it:R1} the $t$-th pair is itself
determined by $\mathcal F_t$.  The number of terminals charged in this case is therefore at most
$\sum_{t\in \mathcal{T}}\mathbf 1\{u_t\in S,\ b_t(u_t)<\theta K\}$, whose expectation is 
$\sum_{t\in \mathcal{T}}\Prob\bigl[u_t\in S,\ b_t(u_t)<\theta K\ \big|\ \mathcal F_t\bigr] \le O(d/R)\cdot\Ex[|T|]$
by \eqref{eq:post1}.  Every terminal reaches the count $D$ at most once and the query counts sum to
$2q$, so $|T|\le2q/D$. Since $DR=\theta K$ and $\theta=\tfrac18$, this case
contributes $O\bigl(dq/(DR)\bigr)=O\bigl(dq/(\theta K)\bigr)=O(dq/K)$.

If $b_t\ge\theta K$ comes first, then by \eqref{eq:bbudget} at most $2q/(\theta K)=O(q/K)$ terminals
can ever be charged this way.
\end{proof}

Finally, let $L$ count the queried pairs $\{u,c\}$, at times $t\le\rho$, for which $c\in C_{i,j}$ is a
center, $u$ is a special terminal that is \emph{not} in $W$, and $u$ belongs to the group
$S_{i,j}$. The pair $(i,j)$ is determined by the rows being public.  Not being in $W$ means that such a $u$ has, at the time of the query, no earlier answer $1$,
fewer than $D$ earlier queries, and $b_t(u)<\theta K$, so \eqref{eq:post3} applies and bounds the
conditional probability of each such occurrence by $O(d/K)$. Summing over the at most $q$
queries,
\begin{equation}
 \label{eq:L}
  \Ex[L]= O\!\left(\frac{dq}K\right).
\end{equation}

\section{Few groups have three noticed members}\label{sec:sparsity}

The set $W$ of \cref{lem:W} is chosen by the algorithm using everything it has seen, so we need a
statement that holds for all sets of terminals at once.  Recall that the $R$ partitions \ref{it:C3}
are independent, and that $R=A^{2/3}$.

A pair $(i,j)$ can only become exposed if three members of the group $S_{i,j}$ have
been touched, and the members we have no individual control over are the noticed ones.  What we need
is therefore a bound on the number of groups, counted across all $R$ blocks, that contain three or
more members of a given set $Q$ of special terminals. Moreover, we need it for every $Q$ at once, since
the set $W$ that will be substituted for $Q$ is chosen by the algorithm and is not fixed in advance.
The quantity $Z(Q)$ below counts exactly those groups.

\begin{lemma}\label{lem:sparsity}
For a set $Q\subseteq S$, define $Z(Q):=\#\bigl\{(i,j)\in[R]\times[A]:\ |S_{i,j}\cap Q|\ge3\bigr\}$.
For all sufficiently large $m$, with probability at least $1-A^{-10}$ the following holds
simultaneously for every $Q\subseteq S$ with $1\le|Q|\le R$:
\[
  Z(Q)\le21\,|Q|\log A .
\]
\end{lemma}
Note that, within one block, the groups are disjoint, so at most $|Q|/3$ pairs with that block can be counted, and $Z(Q)\le R|Q|/3$ always holds. The above lemma replaces the trivial factor $R=A^{2/3}$ by $O(\log A)$.

\begin{proof}
Fix $s$ with $1\le s\le R$, fix $Q$ with $|Q|=s$, and let $z\ge1$.  Call a \emph{witness list} a set
of $z$ distinct pairs $(i_1,Y_1),\dots,(i_z,Y_z)$ in which each $i_p\in[R]$ and each $Y_p$ is a
three-element subset of $Q$, and in which $Y_p\cap Y_{p'}=\varnothing$ whenever $p\ne p'$ and
$i_p=i_{p'}$.  If $Z(Q)\ge z$, pick $z$ pairs $(i,j)$ with $|S_{i,j}\cap Q|\ge3$ and pick three
members of $Q$ inside each. Note that the groups of one block are disjoint, so this is a witness list, and
$Y_p\subseteq S_{i_p,j}$ for some $j$, for every $p$. Hence
\[
  \Prob[Z(Q)\ge z]\ \le\ \sum_{\text{witness lists}}
  \Prob\bigl[\,Y_p\text{ lies in one group of block }i_p\text{ for every }p\,\bigr].
\]
The number of witness lists is at most $\binom{R\binom s3}{z}\le\bigl(eRs^3/(6z)\bigr)^z$.

\smallskip

Fix a witness list and bound its probability.  Write $h_i$ for the number of triples of the list attached to
block $i$, so that $h_1+\dots+h_R=z$.  Fix a block $i$; its $h_i$ triples are pairwise disjoint, so
$3h_i\le s$.  Realize
the partition of block $i$ as a uniformly random one-to-one map $\tau$ from $S$ onto $[A]\times[d]$,
the group of a terminal being the first coordinate of its image.  Reveal $\tau$ on those $h_i$ triples one at a time.  Suppose the first $r-1$ of them have been revealed, so $\tau$ is known on $3(r-1)$
terminals and is a uniformly random one-to-one map from the rest onto the remaining places, and let
the $r$-th triple be $\{y_1,y_2,y_3\}$.  Revealing $\tau(y_1)$ leaves at most $d-1$ free places in its
group out of $dA-3(r-1)-1$ free places overall, and revealing $\tau(y_2)$ in that same group then
leaves at most $d-2$ out of $dA-3(r-1)-2$.  As $3(r-1)+2\le s+2\le R+2$ and $R=A^{2/3}=o(A)$, both
denominators are at least $dA/2$ once $m$ is large, so the conditional probability that
$\{y_1,y_2,y_3\}$ lands in one group is at most $4(d-1)(d-2)/(d^2A^2)=:a_d/A^2$.  Multiplying these $h_i$ conditional
probabilities bounds by $(a_d/A^2)^{h_i}$ the chance that every triple of block $i$ lands in one
group.  The $R$ partitions \ref{it:C3} are independent, so for a fixed witness list
\[
  \Prob\bigl[\,\text{every }Y_p\text{ lands in one group of block }i_p\,\bigr]
  \ \le\ \prod_{i=1}^{R}\Bigl(\frac{a_d}{A^2}\Bigr)^{h_i}
  \ =\ \Bigl(\frac{a_d}{A^2}\Bigr)^{h_1+\dots+h_R}
  \ =\ \Bigl(\frac{a_d}{A^2}\Bigr)^{z}.
\]  Therefore
\[
  \Prob[Z(Q)\ge z]\le\left(\frac{eRs^3}{6z}\cdot\frac{a_d}{A^2}\right)^{z}
  =\left(\frac{a'_d\,Rs^3}{A^2z}\right)^{z},
  \qquad a'_d:=\frac{e\,a_d}6 .
\]
Both constants are absolute: $a_d=4(d-1)(d-2)/d^2<4$, so $a'_d<2e/3$.
Since $s\le R=A^{2/3}$ we have $Rs^3\le A^{2/3}\cdot A^{4/3}\cdot s=A^2s$, so $\Prob[Z(Q)\ge z]\le\left(\frac{a'_d\,s}{z}\right)^{z}$.

\medskip
Now set $z_s:=\lceil 20s\log A\rceil$. For all sufficiently large $m$ such that $\log A\ge e\,a'_d/20$, we have $a'_d s/z_s\le a'_d/(20\log A)\le e^{-1}$. Therefore, $\Pr[Z(Q)\ge z_s]\le e^{-z_s}\le e^{-20s\log A}=A^{-20s}$. For a fixed $s$, the number of possible sets $Q$ of size $s$ is at most $\binom{dA}{s}\le(dA)^s\le A^{2s}$, where the last inequality uses $A\ge d$. Hence, the probability that there exists a set $Q$ of size $s$ with $Z(Q)\ge z_s$ is at most $A^{-18s}$. Taking a union bound over $1\le s\le R$ gives $\sum_{s\ge1}A^{-18s}\le2A^{-18}\le A^{-10}$. Finally, for $A\ge3$ we have $z_s\le21s\log A$. Thus, every such set $Q$ satisfies $Z(Q)<z_s\le21|Q|\log A$.
\end{proof}

The words ``simultaneously for every $Q$'' is what makes this usable: they allow $Q$ to be the set
$W$, even though $W$ depends on the whole transcript.

\section{Proof of the theorem}\label{sec:proof}
This section puts the pieces together.  An algorithm that estimates the optimal cost within a factor
$2-\eta$ must tell $\Dno$ from $\Dyes$, since by \eqref{eq:gap} the two optimal costs are further apart than
that. By \cref{lem:coupling}, it can do so only when the hidden block $g$ is exposed, and
\cref{sec:counting,sec:sparsity} show that $q$ queries leave all but a few of the $R$ blocks unexposed.  Since $g$ is drawn at random and leaves no trace in a \no{} configuration, it is unlikely to be
one of those few, and this factor $1/R$ is what turns the counts of the last two sections into the
exponent $7/5$.  \Cref{lem:tv} carries this out, and the theorem follows by taking $q$ small enough
to keep the two transcripts close.
\begin{lemma}
\label{lem:tv}
For every deterministic algorithm making at most $q$ queries and every sufficiently large $m$,
\[
  \dtv\bigl(\trans(\mathcal A,\Dno),\trans(\mathcal A,\Dyes)\bigr)
  \ \le\ C_d\left(\frac{q\log A}{KR}+\frac q{AK}\right)+A^{-10}
\]
for a constant $C_d$ depending only on $d$.
\end{lemma}

\begin{proof}
By \cref{lem:coupling} it suffices to bound the probability that block $g$ is exposed in \no{} instances.  Let us work on the event $\mathcal R^{\mathsf c}$ (i.e., $\rho > q$), so that all the quantities of \cref{sec:counting} refer to the whole run.

First, the number of exposed pairs $(i,j)$ is at most $Z(W)+L$.  Indeed, let $(i,j)$ be exposed, so
$|\touch_i(j)|\ge3$.  If at least three members of $\touch_i(j)$ lie in $W$ then
$|S_{i,j}\cap W|\ge3$ and $(i,j)$ is counted by $Z(W)$.  Otherwise some $u\in\touch_i(j)$ lies
outside $W$. By \eqref{eq:touch} the algorithm queried a pair $\{u,c\}$ with $c$ a center of row $j$
of block $i$, and $u\in S_{i,j}$, so this queried pair is counted by $L$.  Distinct
exposed pairs are charged to distinct queried pairs, because every center lies in exactly one row.  Since an exposed block contains an exposed pair, the number of exposed blocks is at most
$Z(W)+L$ as well.

Second, $g$ is drawn in \ref{it:C5} independently of everything else, and the \no{} metric does not
involve it, so $g$ is independent of the whole \no{} instance and hence of $W$, $L$ and $Z$.
Conditioning on the instance,
\[
  \Prob[\,g\text{ exposed}\mid\text{the instance}\,]=\frac{\#\{\text{exposed blocks}\}}R
  \le\frac{Z(W)+L}R .
\]
Let $\mathcal S$ be the event of \cref{lem:sparsity}. Write $\mathcal B:=\mathcal R^{\mathsf c}\cap\mathcal S\cap\{|W|\le R\}$, and split $\Prob[\,g\text{ exposed}\,]\ \le\ \Prob[\,g\text{ exposed},\ \mathcal B\,]+\Prob[\mathcal B^{\mathsf c}]$.
For the second term, a union bound gives
$\Prob[\mathcal B^{\mathsf c}]\le\Prob[\mathcal R]+\Prob[\mathcal S^{\mathsf c}]+\Prob[\,|W|>R\,]$.
For the first, note that $\mathcal B$ is determined by the \no{} instance and the run on it. Hence
is independent of $g$, and so we may keep the conditioning on the instance:
on $\mathcal B$ we have $Z(W)\le21|W|\log A$ by \cref{lem:sparsity} applied with $Q=W$, which is
legitimate because $|W|\le R$ there (and trivial if $W=\varnothing$, since then $Z(W)=0$).  Hence
\[
  \Prob[\,g\text{ exposed},\ \mathcal B\,]
  \ \le\ \Ex\Bigl[\mathbf 1_{\mathcal B}\,\frac{Z(W)+L}R\Bigr]
  \ \le\ \frac{21\log A}R\,\Ex[|W|]+\frac{\Ex [L]}R ,
\]
the last step dropping the indicator.  Putting the two halves together,
\[
  \Prob[\,g\text{ exposed}\,]
  \le\frac{21\log A}R\,\Ex[|W|]+\frac{\Ex [L]}R
     +\Prob[\mathcal R]+\Prob[\mathcal S^{\mathsf c}]+\Prob[\,|W|>R\,].
\]
By \cref{lem:W} the first term is $O\bigl(dq\log A/(KR)\bigr)$ and by \eqref{eq:L} the second is $O\bigl(dq/(KR)\bigr)$. By \cref{lem:cluster} the third is $O\bigl(q/(AK)\bigr)$, and the fourth is at
most $A^{-10}$. Finally by \cref{lem:W} and Markov's inequality the fifth is at most
$\Ex[|W|]/R=O\bigl(dq/(KR)\bigr)$. Collecting terms bounds
$\Prob_{\Dno}[\,g\text{ exposed}\,]$ by the right-hand side of the lemma, and \cref{lem:coupling}
turns that into the stated bound on the distance between the two transcripts.
\end{proof}

With $A=m^3$, $R=m^2$ and $K=m^5$ we have
\begin{equation}\label{eq:balance}
  KR=m^7=K^{7/5},\qquad AK=m^8,\qquad \log A=3\log m .
\end{equation}

\begin{proof}[Proof of \cref{thm:main}]
Set $c'_\eta:=1/(12C_d)$ with $C_d$ as in \cref{lem:tv}, and let $q\le c'_\eta\,\frac{KR}{\log A}$.
Then the first term of \cref{lem:tv} is at most $1/12$, while the second is at most
$C_dc'_\eta R/(A\log A)$ and the third is $A^{-10}$. By \eqref{eq:balance} both tend to $0$, so each
is at most $1/12$ for large $m$.  Hence there is an $m'_\eta$ such that
\begin{equation}\label{eq:tvsmall}
  \dtv\bigl(\trans(\mathcal A,\Dno),\trans(\mathcal A,\Dyes)\bigr)\le\tfrac14
  \qquad\text{for all }m\ge m'_\eta
\end{equation}
for every deterministic algorithm $\mathcal A$ making at most $q$ queries.

\medskip
\emph{Distinguishing is impossible.}  Let $\mu:=\tfrac12\Dno+\tfrac12\Dyes$ and let $m\ge m'_\eta$.
By \ref{it:R2} the output of a deterministic algorithm is a function of its transcript, so by
\eqref{eq:tvsmall} and \cref{prop:lecam} every deterministic algorithm making at most $q$ queries
guesses correctly which of the two distributions produced its input with probability at most
$(1+\tfrac14)/2=\tfrac58$.

\emph{Estimating is impossible.}  Let $m\ge\max\{m_\eta,m'_\eta\}$ with $m_\eta$ as in
\eqref{eq:gap}, and let $\mathcal A$ be a deterministic algorithm making at most $q$ queries.
Follow $\mathcal A$ by the rule that answers \yes{} if its output $\widehat C$ satisfies
$\widehat C\le(2-\eta)\Theta_{\mathsf y}$ and \no{} otherwise. The threshold is a number fixed in
advance by \eqref{eq:optyes}, so this rule costs no queries.  On a \yes{} instance an output with
$\st\le\widehat C\le(2-\eta)\,\st$ satisfies $\widehat C\le(2-\eta)\Theta_{\mathsf y}$, and on a
\no{} instance such an output satisfies $\widehat C\ge\Theta_{\mathsf n}>(2-\eta)\Theta_{\mathsf y}$
by \eqref{eq:gap}.  So the composed algorithm guesses correctly whenever $\widehat C$ lies in the
required interval.  By the previous paragraph the composed algorithm is correct with probability at
most $\tfrac58$ under $\mu$, so $\mathcal A$ produces a number in the required interval with
probability at most $\tfrac58$, and fails with probability at least $\tfrac38>\tfrac13$.

\medskip
\emph{From one size to all sizes.}  So far, the point set has
$N_m:=(b+1)m^5+dm^3$ points, by \eqref{eq:ncount}.  Let $n$ be large and let $m$ be the largest
integer with $N_m\le n$. Since $N_{m+1}/N_m\to1$ we have $N_m=\Theta_\eta(n)$, and $m\to\infty$ as
$n\to\infty$.  Add $n-N_m$ further non-terminal points, publicly named, and put every distance
between one of them and another point equal to $2$, keeping $0$ on the diagonal.  The result is
still a metric, since every distance is now at most $2$ and at least the pairwise distances that were already there, and any two distances involving a new point sum to at least $2$.

Note that the optimal solution does not change. A new point is not a terminal, so it may be omitted, and omitting it
never costs anything: in a tree using it, a new point of degree $1$ can simply be deleted, and one of degree $r\ge2$ contributes $2r$ to the weight and may be deleted after joining its $r$ former neighbors by a path of $r-1$ edges of weight at most $2$ each, at a cost of at most $2r-2$.
Repeating removes all new points without increasing the weight, so both $\Theta_{\mathsf n}$ and $\Theta_{\mathsf y}$, and the gap \eqref{eq:gap}, are unchanged. Moreover any algorithm running on
the enlarged instance can be run on the original one at no extra cost, by answering $2$ to every query that involves a new point and passing on the others.  So the bound \eqref{eq:tvsmall} carries
over.

By \eqref{eq:balance} and $K=\Theta_\eta(N_m)=\Theta_\eta(n)$ we have
$KR/\log A=\Theta_\eta\bigl(n^{7/5}/\log n\bigr)$, so shrinking the constant to some
$c_\eta>0$ makes $c_\eta n^{7/5}/\log n\le c'_\eta KR/\log A$ for every large $n$.  This proves the
statement for deterministic algorithms.  The instances have $|T|=K+dA=\Theta_\eta(n)$ terminals by
\eqref{eq:ncount}, and every distance between two distinct points is $\lambda=K^{-3}$, $1$ or $2$,
with $\lambda=\Theta_\eta(n^{-3})$.  Finally \cref{prop:yao} turns the statement for deterministic
algorithms into the statement for randomized ones.
\end{proof}

\section*{AI Disclosure}
A lower bound of $\widetilde{\Omega}(n^{4/3})$, improving upon~\cite{CKT23}, was initially discovered using a custom agentic harness developed by the authors, which incorporated GPT-5.6 Sol and Opus 5. The resulting proof followed, at a high level, ideas similar to the lower-bound construction of~\cite{CKT23}. Subsequently, through interaction with GPT-5.6 Sol, we developed a substantially different construction, which appears in the main theorem of this paper and is further generalized in~\cref{sec:tradeoff}.

The human authors independently verified the arguments, simplified the proofs, and rewrote the presentation. AI tools were also used to assist with drafting and language polishing. The authors take full responsibility for the content and correctness of the paper.

\bibliographystyle{alpha}
\bibliography{ref}
\appendix
\section{The approximation--query trade-off}\label{sec:tradeoff}

The construction of \cref{sec:constr} is the case $\alpha=2$ of a family indexed by an integer
$\alpha\ge2$: its two neighbor-set lists agree on every set of at most two terminals of a group,
and the family replaces two by $\alpha$.  Raising $\alpha$ raises the exponent of the lower bound,
because an algorithm must now make $\alpha+1$ queries inside a single group before it sees anything. 
However, it also lowers the approximation factor that the construction can separate. 
In what follows, we outline how the proof of~\cref{thm:main} can be adapted to show that for any integer parameter $\alpha\ge 2$, estimating the Steiner tree cost within a factor of $\frac{\alpha}{\alpha-1}$ requires $\widetilde{\Omega}\left(n^{(3\alpha+1)/(2\alpha+1)}\right)$ queries (see~\cref{thm:hierarchy,thm:ceiling}).

\subsection{Local lists of higher order}

Let $U$ be a group of $d>\alpha$ special terminals.  A \emph{local list} $F$ on $U$ is a multiset of
subsets of $U$, listed in a fixed order as its $b$ entries; the neighbor sets \eqref{eq:nbn} and
\eqref{eq:nby} of \cref{sec:constr} are the case $\alpha=2$.  For $S\subseteq U$ write
\[
  f_F(S):=\#\{B\in F:\ S\subseteq B\},
  \qquad
  \tr_H(F):=\{\!\{\,B\cap H\ :\ B\in F\,\}\!\},
\]
so that $f_F$ records containments and $\tr_H(F)$ is the multiset of traces on a panel $H\subseteq U$.
By inclusion--exclusion, for every $Z\subseteq H$,
\begin{equation}\label{eq:mobius}
  \#\{B\in F:\ B\cap H=Z\}=\sum_{Z\subseteq X\subseteq H}(-1)^{|X|-|Z|}f_F(X),
\end{equation}
so two lists have equal traces on every panel of size at most $\alpha$ precisely when their
containment counts agree on every set of size at most $\alpha$.  Call such a pair
\emph{$\alpha$-safe}.

\begin{lemma}
\label{lem:mobius}
Fix $\alpha\ge2$ and $d>\alpha$, and put
$a_j:=(-1)^{\alpha-j}\binom{d-j-1}{\alpha-j}$ for $0\le j\le\alpha$.  Let $\mathsf N_\alpha$ contain
$a_j$ copies of every $j$-subset of $U$ for each $j$ with $a_j>0$, and let $\mathsf Y_\alpha$ contain
$-a_j$ copies of every $j$-subset for each $j$ with $a_j<0$, together with one copy of $U$ itself.
Then $\mathsf N_\alpha$ and $\mathsf Y_\alpha$ have the same number $b$ of entries and are
$\alpha$-safe, while no support of $\mathsf N_\alpha$ has more than $\alpha$ elements and
$\mathsf Y_\alpha$ contains a support of size $d$.
\end{lemma}

\begin{proof}
The coefficients $a_j$ are chosen so that, for every $H\subseteq U$ with $|H|\le\alpha$, the signed number of supports containing $H$ is exactly $1$; this will cancel the single copy of $U$ in $\mathsf Y_\alpha$.

Fix such an $H$, and let $h:=|H|$. For each $j\ge h$, exactly $\binom{d-h}{j-h}$ of the $j$-subsets of $U$ contain $H$. Hence
$f_{\mathsf N_\alpha}(H)-f_{\mathsf Y_\alpha}(H)
=\sum_{j=h}^{\alpha}a_j\binom{d-h}{j-h}-1$.
We claim that this sum is $1$. Set $m:=d-h$, $q:=\alpha-h$, and $r:=j-h$, so $q<m$. Then the sum becomes $\sum_{r=0}^{q}\binom mr(-1)^{q-r}\binom{m-r-1}{q-r}$.

By the binomial theorem, $(-1)^{q-r}\binom{m-r-1}{q-r}$ is the coefficient of $x^{q-r}$ in $(1-x)^{m-r-1}$. Thus the $r$-th summand above is the coefficient of $x^q$ in $\binom mr x^r(1-x)^{m-r-1}$. Since $q<m$, terms with $r>q$ cannot contribute to the coefficient of $x^q$, so our sum is the coefficient of $x^q$ in $\sum_{r=0}^{m-1}\binom mr x^r(1-x)^{m-r-1}$.
But $\sum_{r=0}^{m-1}\binom mr x^r(1-x)^{m-r-1}=\frac{(x+(1-x))^m-x^m}{1-x}=\frac{1-x^m}{1-x}=1+x+\cdots+x^{m-1}$. Since $q<m$, the coefficient of $x^q$ in this polynomial is $1$, proving the claim.

Thus $f_{\mathsf N_\alpha}(H)=f_{\mathsf Y_\alpha}(H)$ for every $H$ with $|H|\le\alpha$. By \eqref{eq:mobius}, the two lists therefore have identical traces on every panel of size at most $\alpha$, so they are $\alpha$-safe. Taking $H=\varnothing$ also shows that they contain the same number $b$ of entries.

Finally, every support appearing in $\mathsf N_\alpha$ has size at most $\alpha$ by construction, whereas $\mathsf Y_\alpha$ contains the support $U$, which has size $d$.
\end{proof}
 
For $\alpha=2$ the lemma returns the pair used throughout the paper: $\mathsf N_2$ consists of the
$\binom d2$ pairs together with $\binom{d-1}2$ empty supports and $\mathsf Y_2$ of one full support
together with $d-2$ copies of each singleton, both with $b=(d-1)^2$ entries.
For $\alpha=3$ the roles of the sizes shift by one: $\mathsf N_3$ consists of every $3$-subset
once together with $\binom{d-2}2$ copies of each singleton, and $\mathsf Y_3$ of one full support,
$d-3$ copies of each pair and $\binom{d-1}3$ empty supports, both with $b=d(d-2)(2d-5)/3$ entries.
Here, even the \no{} world has centers serving three terminals, so its profit is no longer zero ($\beta_3=\tfrac12$) and the factor the construction can separate falls from $2$ to $3/2$.
A list of order
$\alpha=1$ would give the matching-versus-star pair, in which a single probe already separates the two worlds.

\subsection{The hierarchy}
We keep the architecture of \cref{sec:constr}. Ordinary terminals are in $A$ hidden clusters of size $R$, special terminals regrouped independently in each of $R$ blocks, one hidden block, and
distances in $\{\lambda,1,2\}$ and replace the \yes{} and \no{} blocks by the described $\mathsf N_\alpha$ and $\mathsf Y_\alpha$.  Say that a group is \emph{exposed} when $\alpha+1$ of its
members have been probed at centers of its own row.  Every step of \cref{sec:trade,sec:coupling} goes
through with $\alpha$ in place of $2$: \cref{lem:mobius} supplies, for each panel $H$ of size at most
$\alpha$, a permutation of $[b]$ matching the entries of the two lists so as to preserve all incidences with $H$, and the
injection of \cref{lem:coupling} is built from those permutations.

As in \cref{sec:constr}, a configuration $\omega$ determines two instances on the same point
set: the \no{} instance $I_{\mathsf n}(\omega)$, which uses $\mathsf N_\alpha$ in every block, and the
\yes{} instance $I_{\mathsf y}(\omega)$, which uses $\mathsf Y_\alpha$ in block $g$ and
$\mathsf N_\alpha$ elsewhere.  For the rest of this section $\Dno$ and $\Dyes$ denote their
distributions, $\omega$ being uniform; for $\alpha=2$ these are the instances and distributions of
\cref{sec:constr}.

\begin{theorem}
\label{thm:hierarchy}
Fix an integer $\alpha\ge2$ and a constant $c<\alpha/(\alpha-1)$.  For infinitely many $n$ these
distributions live on $n$-point instances and satisfy
\[
  c\cdot\st(I_{\mathsf y})<\st(I_{\mathsf n})
  \qquad\text{for all }I_{\mathsf y}\in\operatorname{supp}(\Dyes),\
  I_{\mathsf n}\in\operatorname{supp}(\Dno),
\]
while every randomized algorithm that estimates the metric Steiner tree cost within a factor
$c$ with probability at least $2/3$ makes $\Omega\!\left(\frac{n^{e_\alpha}}{\log n}\right)$ queries, where $e_\alpha:=\frac{3\alpha+1}{2\alpha+1}$.
\end{theorem}

The exponent comes from the parameters $A=m^{\alpha+1}$, $R=m^{\alpha}$, $K=AR=m^{2\alpha+1}$, for
which $n=\Theta(K)$ and $KR/\log A=\Theta(n^{(3\alpha+1)/(2\alpha+1)}/\log n)$; the case $\alpha=2$
is \cref{thm:main}.  The separation of the costs is the computation of \cref{sec:gap} carried out
for general $\alpha$.  Writing $\beta_2:=0$ and $\beta_\alpha:=(\alpha-2)/(\alpha-1)$ for
$\alpha\ge3$, a $\mathsf N_\alpha$ component with $t$ special terminals and $s$ centers satisfies
$t+s-1\le\alpha s$ and hence has profit $t-s-1\le\beta_\alpha(t-1)$, so the profit available in the \no{} world is at most $\beta_\alpha(dA-1)$, whereas the $A$ full stars of the hidden block give the \yes{} world profit at least $A(d-2)$. Hence, the ratio of the two optimal costs tends to $2-\beta_\alpha=\alpha/(\alpha-1)$.

\subsection{A ceiling for exact trace matching}

The ratio $\alpha/(\alpha-1)$ is not merely what the above computation delivers; it is a barrier for
the whole architecture, and it is what makes the trade-off unavoidable.

\begin{theorem}
\label{thm:ceiling}
Let two local lists be $\alpha$-safe in the sense above, and place them in the single-scale
architecture of \cref{sec:constr}.  Then the optimal costs of the two resulting instances differ by a factor
strictly smaller than $\alpha/(\alpha-1)$.
\end{theorem}

\begin{proof}
Since $f_F(H)$ is the multiplicity of the full trace $H$ on the panel $H$, $\alpha$-safety gives the
coverage property: if $|H|\le\alpha$ and $H\subseteq B$ for some support $B$ of one rule, then
$H\subseteq B'$ for some support $B'$ of the other.

Consider an optimal incidence forest for whichever of the two systems has the larger profit, and let
$p$ be that profit.  Root each nontrivial component at a terminal.  A center with one parent and
$t\ge1$ terminal children contributes $t-1$ to $p$.  Split those children into
$r_t=\lceil t/(\alpha-1)\rceil$ groups of size at most $\alpha-1$; adjoining the parent turns each
group into a certificate of size at most $\alpha$, so coverage supplies, on the other side, a support
containing it.  Selecting one such center per certificate and merging repetitions yields an incidence
graph that connects every pair the original forest connected, uses at most $\sum_t r_t$ centers and
has no more components than the original.  A spanning forest of it therefore has profit at least
\[
  \sum_t (t-r_t)\ \ge\ \frac{\alpha-2}{\alpha-1}\sum_t (t-1)\ =\ \beta_\alpha p ,
\]
the inequality being the elementary bound $t-\lceil t/(\alpha-1)\rceil\ge\beta_\alpha(t-1)$, checked
by writing $t$ modulo $\alpha-1$.  So if $p\ge q$ are the two maximum profits then
$q\ge\beta_\alpha p$.

Write $B:=L+2(\kappa-1)$ for the common part of the cost identity of \cref{sec:gap}, so that each
optimal cost equals $B$ minus the corresponding maximum profit.  A nonempty incidence forest has profit at
most $\kappa-2$, so $B-2p>0$, and therefore
\[
  \frac{\st_{\mathrm{high}}}{\st_{\mathrm{low}}}
  =\frac{B-q}{B-p}
  \ \le\ \frac{B-\beta_\alpha p}{B-p}
  \ <\ 2-\beta_\alpha=\frac{\alpha}{\alpha-1}. \qedhere
\]
\end{proof}

\end{document}